\documentclass[twocolumn,conference,letterpaper]{IEEEtran}
\usepackage[left=0.75in,right=0.75in,top=1in,bottom=0.75in]{geometry}

\usepackage{cite}
\usepackage{amssymb,amsmath,latexsym,amsfonts,amsthm,mathtools}
\usepackage[colorlinks=true]{hyperref}
\hypersetup{colorlinks,breaklinks,citecolor=blue,linkcolor=blue,urlcolor=blue}
\usepackage{graphicx}
\usepackage{booktabs}
\usepackage{array}
\usepackage{xcolor}
\usepackage{placeins}
\usepackage[nameinlink]{cleveref}
\usepackage{tikz}
\usepackage{pgfplots}
\usepackage{graphicx}
\usepackage{subcaption}
\usepgfplotslibrary{groupplots}
\pgfplotsset{compat=1.18}
\usetikzlibrary{arrows.meta,calc,positioning}

\definecolor{designblue}{RGB}{32,92,160}
\definecolor{designorange}{RGB}{210,112,32}
\definecolor{designgreen}{RGB}{42,132,90}
\definecolor{designgray}{RGB}{90,96,105}
\definecolor{designred}{RGB}{175,55,55}

\theoremstyle{plain}
\newtheorem{theorem}{Theorem}
\newtheorem{lemma}{Lemma}
\newtheorem{corollary}{Corollary}

\newtheorem*{problem*}{Problem}
\theoremstyle{remark}
\newtheorem{remark}{Remark}
\newtheorem{assumption}{Assumption}
\theoremstyle{definition}

\DeclareMathOperator{\dist}{dist}

\newcommand{\R}{\mathbb{R}}
\newcommand{\calC}{\mathcal{C}}
\newcommand{\calD}{\mathcal{D}}
\newcommand{\calF}{\mathcal{F}}
\newcommand{\calM}{\mathcal{M}}
\newcommand{\calT}{\mathcal{T}}

\newcommand{\abs}[1]{\left\lvert #1 \right\rvert}
\newcommand{\trans}{^{\mathsf{T}}}

\IEEEoverridecommandlockouts

\begin{document}

\title{Impact-Time Guidance via Normal Contraction to a Time-to-Go Isochron}

\author{Shivam Bajpai and Abhinav Sinha,~\IEEEmembership{Senior Member,~IEEE}
\thanks{ The authors are with the Guidance, Autonomy, Learning, and Control for Intelligent Systems (GALACxIS) Lab, Department of Aerospace Engineering and Engineering Mechanics, University of Cincinnati, Cincinnati, OH, 45221,  USA
 e-mails: bajpaism@mail.uc.edu, abhinav.sinha@uc.edu}
    }	

\maketitle
\thispagestyle{empty}

\begin{abstract}
We develop a contraction-based perspective on impact-time guidance that augments a baseline homing command with a timing bias. The proposed perspective treats the prescribed schedule as a moving time-to-go isochron and regulates motion normal to that set through velocity-normal lateral acceleration while the interceptor's speed remains constant. We derive a transport equation that characterizes homing-compatible time-to-go coordinates and define a predictor defect that quantifies the mismatch of approximate maps. We show that the scalar timing channel induces a coordinate-invariant rank-one metric on the normal quotient. To account for bounded lateral acceleration, we formulate a robust scalar filter and derive a necessary and sufficient condition for pointwise feasibility. We then show that terminal calibration and funnel invariance establish first interception at the prescribed time under the stated assumptions. We also develop a preterminal alignment and homing handover that avoids singular inversion as lateral timing authority vanishes near collision-course alignment. The proposed perspective accommodates analytic, numerical, and learned time-to-go maps that satisfy the required calibration and regularity conditions.
	\end{abstract}
	
\begin{IEEEkeywords}
Contraction analysis, impact-time guidance, lateral acceleration, intercept guidance, time-to-go.
\end{IEEEkeywords}

\section{Introduction}
Impact-time guidance addresses a homing task in which the interceptor must capture a target at an assigned instant \cite{jeon2006impact}. Coordinating arrival times facilitates simultaneous and cooperative interception \cite{sinha2020super,sinha2021event}, as well as time-constrained pursuit-evasion \cite{sinha2021aircraft}. Related scheduling objectives arise in time-constrained spacecraft rendezvous \cite{xie2026exact}, spacecraft docking \cite{zhang2022appointed}, and scheduled aircraft arrivals \cite{murrietamendoza2020arrival}.

A natural way to incorporate the schedule is to augment a baseline homing command with a timing correction, as in biased proportional navigation (PN) \cite{jeon2006impact}. Lyapunov-based and sliding-mode designs provide nonlinear timing regulation; see \cite{kim2015lyapunov,kumar2015large,sinha2020super}, while formulations based on true PN and deviated pursuit use time-to-go expressions associated with their respective homing laws \cite{kumar2022true,sinha2025deviated}. The predictor generally links the regulated timing error to an estimate of the remaining flight duration. Extending a convergence argument to another time-to-go map requires checking how the predicted time evolves under the executing homing law and whether timing convergence corresponds to physical arrival. The schedule or the prescribed impact time must also remain compatible with available acceleration \cite{gopikannan2026bounded}. For velocity-normal steering, lateral timing authority depends on the engagement geometry and may diminish as collision-course alignment develops in the endgame. The dynamical analysis of the timing channel must accommodate declining authority and establish first capture at the assigned instant.

However, states with different ranges, headings, and future paths can share the same time-to-go, so an assigned arrival time selects a moving level set. Tangential perturbations preserve the timing coordinate to first order, while normal perturbations change the timing error. Contraction analysis and differential Lyapunov theory describe convergence through infinitesimal perturbations \cite{lohmiller1998contraction,forni2014differential}; related developments address orbital stability through transverse contraction \cite{manchester2014transverse} and nonlinear feedback synthesis through control contraction metrics \cite{manchester2017ccm}. For impact-time guidance, the contraction viewpoint relates the nominal timing-error decay established by a scalar Lyapunov argument to local convergence toward the scheduled set. A coordinate-invariant normal metric measures variations that affect timing and vanishes on tangent directions, leaving tangential homing motion unrestricted by the convergence requirement. This geometric description makes the convergence statement independent of the chosen engagement coordinates and retains the freedom to achieve the same arrival schedule along different homing trajectories.

Timing-error convergence must be achieved within the interceptor's acceleration limit. Prescribed-time feedback assigns a convergence deadline \cite{song2017prescribed}, while admissibility-preserving input realization (APIR) embeds actuator bounds in the input dynamics and establishes tracking under compatibility conditions on the desired motion and available control authority \cite{kumar2026admissibility}. For interception, the time-to-go coordinate must be calibrated so that convergence of the coordinate to zero corresponds to physical capture. Keeping the coordinate positive before the assigned instant excludes early interception. The remaining challenge is to connect timing convergence to first capture under bounded acceleration while accounting for predictor mismatch and declining timing authority. We develop a moving time-to-go isochron formulation with the following contributions:
\begin{itemize}
    \item We characterize homing-compatible maps through a transport boundary-value problem and define a predictor defect for approximate maps. The scalar timing channel induces a coordinate-invariant rank-one metric for contraction on the normal quotient.
    \item We derive a necessary and sufficient feasibility condition for a robust scalar projection under the acceleration bound. Terminal calibration connects funnel invariance to first capture at the assigned instant.
    \item We establish preterminal alignment and homing handover conditions that preserve the assigned arrival time for an exact nominal predictor as timing authority diminishes.
\end{itemize}
For a PN baseline, the command retains the biased-PN structure. A stationary-target PN construction and four alternative predictors illustrate the analysis under velocity-normal actuation. Under the stated coordinate, internal-dynamics, and feasibility assumptions, the proposed guidance law ensures first capture at the assigned instant within the interceptor's acceleration limit.

\section{Problem Formulation}\label{sec:problem}
Let $\mathbf{p}_{\mathrm{P}},\mathbf{p}_{\mathrm{T}}\in\R^{2}$ denote the interceptor's and target's positions, respectively. The interceptor has constant speed $V_{\mathrm{P}}>0$, flight-path angle $\gamma_{\mathrm{P}}$, velocity unit vector $\mathbf{h}_{\mathrm{P}}=[\cos{\gamma_{\mathrm{P}}},\sin{\gamma_{\mathrm{P}}}]\trans$, and scalar lateral acceleration $a_{\mathrm{P}}$. The interceptor's kinematics satisfy
\begin{align}
    \dot{\mathbf{p}}_{\mathrm{P}}
    =
    V_{\mathrm{P}}\mathbf{h}_{\mathrm{P}},
    ~
    \dot{\gamma}_{\mathrm{P}}
    =
    \frac{a_{\mathrm{P}}}{V_{\mathrm{P}}};
    ~
    \abs{a_{\mathrm{P}}}
    \leq
    a_{\max}.
    \label{eq:interceptor_kinematics}
\end{align}
Thus, $a_{\mathrm{P}}$ is the interceptor's sole steering command and changes the velocity direction while a propulsion loop maintains $V_{\mathrm{P}}$. An independent radial channel is absent. The kinematic model presumes that a sufficiently fast acceleration autopilot tracks $a_{\mathrm{P}}$.

Write the relative position as $\mathbf{r}=\mathbf{p}_{\mathrm{T}}-\mathbf{p}_{\mathrm{P}}=r\mathbf{e}_r$. Let $\theta$ denote the line-of-sight (LOS) angle, with $\mathbf{e}_{\theta}$ the counterclockwise normal to $\mathbf{e}_r$. We define $\sigma
    \coloneqq
    \gamma_{\mathrm{P}}-\theta,
    ~
    V_r
    \coloneqq
    \dot{r},
    ~
    V_{\theta}
    \coloneqq
    r\dot{\theta}$. With $a_{\mathrm{T}r}=\mathbf{e}_r\trans\mathbf{a}_{\mathrm{T}}$ and $a_{\mathrm{T}\theta}=\mathbf{e}_{\theta}\trans\mathbf{a}_{\mathrm{T}}$, the polar kinematics are
\begin{align}
    \dot{r}
    =&
    V_r,
    ~
    \dot{\theta}
    =
    \frac{V_{\theta}}{r},
    \label{eq:polar_1}
    \\
    \dot{V}_r
    =&
    \frac{V_{\theta}^2}{r}
    +a_{\mathrm{T}r}
    +a_{\mathrm{P}}\sin{\sigma},
    \label{eq:polar_2}
\end{align}
The transverse velocity and heading-error dynamics retain the geometry-dependent effects of the same lateral input, so
\begin{align}
    \dot{V}_{\theta}
    =&
    -\frac{V_rV_{\theta}}{r}
    +a_{\mathrm{T}\theta}
    -a_{\mathrm{P}}\cos{\sigma},
    \label{eq:polar_3}
    \\
    \dot{\sigma}
    =&
    \frac{a_{\mathrm{P}}}{V_{\mathrm{P}}}
    -\frac{V_{\theta}}{r}.
    \label{eq:polar_4}
\end{align}
The polar dynamics \eqref{eq:polar_2}--\eqref{eq:polar_4} account for both the radial and transverse LOS components of velocity-normal acceleration. For analysis, we use the minimal Cartesian state $\mathbf{x}\coloneqq[\mathbf{p}_{\mathrm{P}}\trans,\gamma_{\mathrm{P}},\mathbf{z}_{\mathrm{T}}\trans]\trans$, where $\mathbf{z}_{\mathrm{T}}$ collects the variables describing the target in the selected predictor. For a stationary target, $\mathbf{z}_{\mathrm{T}}=\mathbf{p}_{\mathrm{T}}$; a constant-velocity model uses $\mathbf{z}_{\mathrm{T}}=[\mathbf{p}_{\mathrm{T}}\trans,\mathbf{v}_{\mathrm{T}}\trans]\trans$; and a maneuver model may append acceleration or estimator states. The polar quantities in \eqref{eq:polar_1}--\eqref{eq:polar_4} are functions of $\mathbf{x}$ rather than independent states. The state dynamics take the form
\begin{align}
    \dot{\mathbf{x}}
    =&
    \mathbf{f}(\mathbf{x},t)
    +\mathbf{g}(\mathbf{x},t)a_{\mathrm{P}}
    +\mathbf{d}(\mathbf{x},t).
    \label{eq:affine_plant}
\end{align}
The nominal target model is included in $\mathbf{f}$, while $\mathbf{d}$ contains prediction error and unmodeled components of the target's acceleration. The fields are locally Lipschitz and continuously differentiable wherever the differential analysis is invoked. Let $\calC$ be a closed capture set and define $q(\mathbf{x})=\dist(\mathbf{x},\calC)$.

Consider a baseline homing law $a_{\mathrm{H}}(\mathbf{x},t)$, such as PN, whose solutions remain in an engagement domain $\calD$ and reach $\calC$. Let $\mathbf{x}_{\mathrm{H}}(\tau;t,\mathbf{x})$ denote the trajectory of \eqref{eq:affine_plant} under $a_{\mathrm{P}}=a_{\mathrm{H}}$ and $\mathbf{d}=0$, with $\mathbf{x}_{\mathrm{H}}(t;t,\mathbf{x})=\mathbf{x}$. The baseline first hitting time, measured from the current time, is
\begin{align}
    \calT_{\mathrm{H}}(\mathbf{x},t)\coloneqq&
    \inf\left\{
    \tau\geq0:
    \mathbf{x}_{\mathrm{H}}(t+\tau;t,\mathbf{x})\in\calC
    \right\}.
    \label{eq:baseline_hitting_time}
\end{align}
Thus, $\calT_{\mathrm{H}}(\mathbf{x},t)$ is the baseline time-to-go, $t+\calT_{\mathrm{H}}(\mathbf{x},t)$ is the corresponding predicted absolute impact time, and $t_f$ denotes the prescribed absolute impact time. Along every nominal homing trajectory, the exact map decreases at one second per second and vanishes at capture. Those two properties yield the transport boundary-value problem
\begin{align}
    \frac{\partial\calT}{\partial t}
    +\nabla_{\mathbf{x}}\calT\trans
    \left(\mathbf{f}+\mathbf{g}a_{\mathrm{H}}\right)
    =&
    -1;
    ~~
    \calT\big|_{\partial\calC}=0,
    \label{eq:transport_pde}
\end{align}
whose characteristic curves are the nominal baseline trajectories $\mathbf{x}_{\mathrm{H}}(\tau;t,\mathbf{x})$. The transport problem can be solved analytically, integrated along the baseline trajectories, or tabulated offline. Closed-form time-to-go approximations remain useful when their disagreement with the baseline flow is quantified. For any candidate map, define the predictor defect
\begin{align}
    \varepsilon_{\calT}
    \coloneqq&
    \frac{\partial\calT}{\partial t}
    +\nabla_{\mathbf{x}}\calT\trans
    \left(\mathbf{f}+\mathbf{g}a_{\mathrm{H}}\right)
    +1.
    \label{eq:predictor_defect}
\end{align}
Hence, $\varepsilon_{\calT}=0$ when the selected map evolves as the exact baseline time-to-go under the nominal homing flow. A nonzero defect identifies the term that must be canceled, bounded, or carried into the timing-error estimate.
\begin{assumption}
\label{ass:coordinate}
The map $\calT:(\calD\setminus\calC)\times\R_{\geq0}\to\R_{>0}$ is twice continuously differentiable and extends continuously to $\calC$ with $\calT=0$ on $\calC$. There exist a terminal neighborhood $\mathcal{N}_{\calC}\subset\calD$ and a constant $\tau_c>0$ such that
\begin{align}
    \{\mathbf{x}\in\calD:0\leq\calT(\mathbf{x},t)\leq\tau_c\}
    \subseteq&\ \mathcal{N}_{\calC},
    \label{eq:terminal_sublevel}
\end{align}
and, on $\mathcal{N}_{\calC}$,
\begin{align}
    \underline{\alpha}(q)
    \leq&
    \calT(\mathbf{x},t)
    \leq
    \overline{\alpha}(q)
    \label{eq:terminal_calibration}
\end{align}
for class-$\mathcal{K}$ functions $\underline{\alpha}$ and
$\overline{\alpha}$. On every compact timing tube considered in the analysis, the level sets of $\calT$ are regular, and there exist constants $0<\underline{p}\leq\overline{p}<\infty$ such that
$\underline{p}\leq\|\nabla_{\mathbf{x}}\calT\|\leq\overline{p}$.
The baseline and timing-corrected internal dynamics remain in $\calD$, and solutions extend continuously to the assigned terminal time.
\end{assumption}
Terminal calibration in \eqref{eq:terminal_calibration} and the terminal-sublevel condition \eqref{eq:terminal_sublevel} tie small coordinate values to proximity to capture. A sufficient construction is
$\calT=q h(\boldsymbol{\eta})$ with
$0<h_-\leq h\leq h_+$ and bounded derivatives. The internal-dynamics requirement in \Cref{ass:coordinate} remains essential: normal contraction regulates timing, while seeker field of view, terminal angle, and other engagement coordinates require their own constraints.

The predictor defect quantifies the accumulated error in the baseline arrival-time estimate.
\begin{lemma}
\label{prop:defect_accumulation}
Let $\calT_{\mathrm{H}}$ be the exact baseline first-hitting-time map in \eqref{eq:baseline_hitting_time}, and let $\calT$ satisfy the same zero boundary condition on $\calC$. Along the nominal baseline trajectory $\mathbf{x}_{\mathrm{H}}(\tau)\coloneqq\mathbf{x}_{\mathrm{H}}(\tau;t,\mathbf{x})$ from $(t,\mathbf{x})$ to first capture,
\begin{align}
\calT(\mathbf{x},t)-\calT_{\mathrm{H}}(\mathbf{x},t)
=&-\int_t^{t+\calT_{\mathrm{H}}(\mathbf{x},t)}
\varepsilon_{\calT}(\mathbf{x}_{\mathrm{H}}(\tau),\tau)\,d\tau.
\label{eq:defect_accumulation}
\end{align}
If $|\varepsilon_{\calT}|\leq\bar{\varepsilon}$ along $\mathbf{x}_{\mathrm{H}}(\tau)$ before capture, then
\begin{align}
|\calT-\calT_{\mathrm{H}}|\leq&\ \bar{\varepsilon}\calT_{\mathrm{H}}.
\label{eq:defect_uniform_bound}
\end{align}
More generally, if $|\varepsilon_{\calT}|\leq c_{\varepsilon}\calT_{\mathrm{H}}^p$ along $\mathbf{x}_{\mathrm{H}}(\tau)$ before capture for some $p\geq0$, then
\begin{align}
|\calT-\calT_{\mathrm{H}}|\leq&\ \frac{c_{\varepsilon}}{p+1}\calT_{\mathrm{H}}^{p+1}.
\label{eq:defect_order_bound}
\end{align}
\end{lemma}
\begin{proof}
Along the nominal baseline trajectory $\mathbf{x}_{\mathrm{H}}(\tau;t,\mathbf{x})$, \eqref{eq:predictor_defect} implies $d\calT/d\tau=-1+\varepsilon_{\calT}$, while \eqref{eq:transport_pde} implies $d\calT_{\mathrm{H}}/d\tau=-1$. The derivative of $\calT-\calT_{\mathrm{H}}$ is thus $\varepsilon_{\calT}$. The common zero boundary condition and the hitting-time definition \eqref{eq:baseline_hitting_time} imply that both coordinates vanish at $t+\calT_{\mathrm{H}}(\mathbf{x},t)$. Integrating the difference dynamics up to capture establishes \eqref{eq:defect_accumulation}. Substitution of $|\varepsilon_{\calT}|\leq\bar{\varepsilon}$ into \eqref{eq:defect_accumulation} yields \eqref{eq:defect_uniform_bound}. For the higher-order estimate, \eqref{eq:transport_pde} implies $\calT_{\mathrm{H}}(\mathbf{x}_{\mathrm{H}}(\tau),\tau)=t+\calT_{\mathrm{H}}(\mathbf{x},t)-\tau$. Integrating the assumed bound $|\varepsilon_{\calT}|\leq c_{\varepsilon}\calT_{\mathrm{H}}^p$ in \eqref{eq:defect_accumulation} then establishes \eqref{eq:defect_order_bound}.
\end{proof}

States with different ranges, headings, and future paths may share the same baseline time-to-go under $a_{\mathrm{H}}$. The scheduling objective is to reach the appropriate moving level set while the homing law governs motion within that set. This geometry motivates the time-to-go isochron formulation. Given a requested impact time $t_f>t_0$, define $s(t)
    \coloneqq
    t_f-t$. The timing error and scheduled level set, respectively, are
\begin{align}
    e(\mathbf{x},t)
    \coloneqq&
    \calT(\mathbf{x},t)-s(t),
    \label{eq:timing_error}
    \\
    \calM_{\mathrm{IT}}(t)
    \coloneqq&
    \left\{
    \mathbf{x}\in\calD\setminus\calC:
    e(\mathbf{x},t)=0
    \right\}.
    \label{eq:impact_manifold}
\end{align}
When $\calT=\calT_{\mathrm{H}}$, every state in $\calM_{\mathrm{IT}}(t)$ reaches $\calC$ at the prescribed absolute time $t_f$ under $a_{\mathrm{H}}$, so $\calM_{\mathrm{IT}}(t)$ is a time-to-go isochron. For an approximate $\calT$, the isochron interpretation requires verification of homing-flow compatibility. The control effectiveness normal to $\calM_{\mathrm{IT}}(t)$ is
\begin{align}
    B_{\calT}
    \coloneqq&
    \nabla_{\mathbf{x}}\calT\trans\mathbf{g}.
    \label{eq:timing_effectiveness}
\end{align}
In a minimal Cartesian state $(\mathbf{p}_{\mathrm{P}},\gamma_{\mathrm{P}},\mathbf{z}_{\mathrm{T}})$,
\begin{align}
    B_{\calT}
    =&
    \frac{1}{V_{\mathrm{P}}}
    \frac{\partial\calT}{\partial\gamma_{\mathrm{P}}}.
    \label{eq:minimal_B}
\end{align}
When redundant polar variables $V_r,V_{\theta}$ are used, their dependence on $\gamma_{\mathrm{P}}$ must be propagated through \eqref{eq:polar_2}--\eqref{eq:polar_3}. Treating $a_{\mathrm{P}}$ as an independent tangential acceleration produces a coefficient inconsistent with the interceptor's kinematic model in \eqref{eq:interceptor_kinematics}.

Along \eqref{eq:affine_plant}, the timing error satisfies the scalar channel
\begin{align}
    \dot{e}
    =&
    \varepsilon_{\calT}
    +B_{\calT}\left(a_{\mathrm{P}}-a_{\mathrm{H}}\right)
    +\Delta_{\calT},
    \label{eq:timing_channel_relative}
\end{align}
where $\Delta_{\calT}
    \coloneqq
    \nabla_{\mathbf{x}}\calT\trans\mathbf{d}$. Equivalently,
\begin{align}
    \dot{e}
    =&
    F_{\calT}+B_{\calT}a_{\mathrm{P}}+\Delta_{\calT},
    \label{eq:timing_channel_total}
\end{align}
with $F_{\calT}
    \coloneqq
    \varepsilon_{\calT}-B_{\calT}a_{\mathrm{H}}$. The relative form \eqref{eq:timing_channel_relative} motivates the command decomposition
\begin{align}
    a_{\mathrm{P}}
    =&
    a_{\mathrm{H}}+a_{\mathrm{IT}},
    \label{eq:biased_homing_decomposition}
\end{align}
where $a_{\mathrm{IT}}$ is the timing bias. The command above has the conventional biased-homing structure. The analysis establishes conditions under which the timing bias compensates for the predictor defect and produces contraction toward $\calM_{\mathrm{IT}}(t)$ while respecting the lateral acceleration limit.
\begin{problem*}
Given $t_f>t_0$, construct a bounded lateral timing bias $a_{\mathrm{IT}}(\mathbf{x},t)$ in \eqref{eq:biased_homing_decomposition} that induces normal contraction toward $\calM_{\mathrm{IT}}(t)$ while satisfying $\abs{a_{\mathrm{H}}+a_{\mathrm{IT}}}\leq a_{\max}$. We also require first capture at $t_f$, namely,
\begin{align}
    \mathbf{x}(t)\notin&\ \calC,~\forall t\in[t_0,t_f),~~
    \mathbf{x}(t_f)\in \calC.
    \label{eq:problem_first_capture}
\end{align}
When $\varepsilon_{\calT}=0$ and $\mathbf{d}=0$, we require $a_{\mathrm{IT}}=0$ on $\calM_{\mathrm{IT}}(t)$, so the baseline homing dynamics are recovered after timing alignment.
\end{problem*}
The acceleration limit constrains the attainable impact times, since the prescribed motion must be compatible with the available lateral authority. Similar feasibility requirements arise in constrained tracking with APIR \cite{kumar2026admissibility} and in safe coordination under time-varying output constraints \cite{sinha2026safeconsensus,sinha2026networked}. For the constant-speed interceptor considered here, the acceleration limit imposes the curvature bound
$\kappa_{\max}=a_{\max}/V_{\mathrm{P}}^2$. Let
$L_{\mathrm{D}}(\mathbf{x}_0,\calC;\kappa_{\max})$ be the shortest bounded-curvature path from the interceptor's initial pose to a stationary capture set. A necessary offline screen is $ t_f-t_0
    \geq
    \dfrac{
    L_{\mathrm{D}}(\mathbf{x}_0,\calC;\kappa_{\max})
    }{V_{\mathrm{P}}}$, which specializes the bounded-curvature geometry in \cite{dubins1957curves}. The inequality provides a preliminary screen for a stationary capture set; moving-target reachability requires further analysis. The online feasibility condition developed next accounts for the timing sensitivity along the engagement.

\section{Main Results}\label{sec:main_results}
We first analyze the unconstrained timing channel, then establish bounded-input feasibility and first capture. The coordinate construction and predictor comparison prepare the terminal analysis, which concludes with preterminal alignment and homing handover. Let $\hat{\Delta}_{\calT}$ be an available compensation term and define
$\tilde{\Delta}_{\calT}
=\Delta_{\calT}-\hat{\Delta}_{\calT}$.
Wherever $B_{\calT}\neq0$, the timing bias cancels the known drift and imposes a prescribed normal rate. The total lateral acceleration takes the form
\begin{align}
    a_{\mathrm{P}}^{\star}
    =&
    a_{\mathrm{H}}
    -\frac{
    \varepsilon_{\calT}
    +\hat{\Delta}_{\calT}
    +k(t)e
    }{B_{\calT}};
    ~~
    k(t)>0.
    \label{eq:ideal_command}
\end{align}
The second term in \eqref{eq:ideal_command} is $a_{\mathrm{IT}}$ in
\eqref{eq:biased_homing_decomposition}. The total command acts through lateral acceleration, with heading changes affecting both LOS velocity components. For an exact nominal coordinate, the timing bias vanishes on the scheduled isochron and the baseline homing law is recovered.
\begin{theorem}
\label{thm:contraction}
Suppose \Cref{ass:coordinate} holds and
\eqref{eq:ideal_command} is well defined on a tubular neighborhood of
$\calM_{\mathrm{IT}}(t)$. The timing error satisfies
\begin{align}
    \dot{e}
    =&
    -k(t)e+\tilde{\Delta}_{\calT}.
    \label{eq:contracting_error}
\end{align}
In the nominal case, let $\mathbf{F}_{\mathrm{cl}}$ be the closed-loop vector field,
$A_{\mathrm{cl}}=\partial\mathbf{F}_{\mathrm{cl}}/\partial\mathbf{x}$,
$\mathbf{p}=\nabla_{\mathbf{x}}e$, and $M_{\perp}
    \coloneqq
    \mathbf{p}\mathbf{p}\trans$. The moving manifold $\calM_{\mathrm{IT}}(t)$ is invariant, and
\begin{align}
    \dot{M}_{\perp}
    +A_{\mathrm{cl}}\trans M_{\perp}
    +M_{\perp}A_{\mathrm{cl}}
    =&
    -2k(t)M_{\perp}.
    \label{eq:metric_identity}
\end{align}
On $\calM_{\mathrm{IT}}(t)$,
$\ker M_{\perp}=T_{\mathbf{x}}\calM_{\mathrm{IT}}(t)$.
The metric $M_{\perp}$ is positive definite on the normal quotient
$T_{\mathbf{x}}\calD/T_{\mathbf{x}}\calM_{\mathrm{IT}}(t)$ and semidefinite on the full state space. If
$\abs{\tilde{\Delta}_{\calT}}\leq\bar{\delta}(t)$, then
\begin{align}
    \abs{e(t)}
    \leq&
    \Phi(t,t_0)\abs{e(t_0)}
    +
    \int_{t_0}^{t}
    \Phi(t,\tau)\bar{\delta}(\tau)\,d\tau,
    \label{eq:robust_contraction_bound}
    \\
    \Phi(t,\tau)
    \coloneqq&
    \exp\left(
    -\int_{\tau}^{t}k(\xi)\,d\xi
    \right).
    \label{eq:transition_function}
\end{align}
\end{theorem}
\begin{proof}
Substituting \eqref{eq:ideal_command} into
\eqref{eq:timing_channel_relative} yields
\eqref{eq:contracting_error}. In the nominal case, \eqref{eq:contracting_error} preserves $e=0$ and implies $\delta\dot{e}=-k(t)\delta e$. Differentiating
$\partial e/\partial t+\nabla e\trans\mathbf{F}_{\mathrm{cl}}=-k(t)e$
with respect to $\mathbf{x}$ produces
$\dot{\mathbf{p}}+A_{\mathrm{cl}}\trans\mathbf{p}=-k\mathbf{p}$,
and hence \eqref{eq:metric_identity}. The level-set regularity in \Cref{ass:coordinate} identifies $\ker M_{\perp}=T_{\mathbf{x}}\calM_{\mathrm{IT}}(t)$.
Under a diffeomorphism with Jacobian $J$,
$M_{\perp}$ transforms as
$J^{-\mathsf{T}}M_{\perp}J^{-1}$, so
$\delta\mathbf{x}\trans M_{\perp}\delta\mathbf{x}=(\delta e)^2$
is coordinate invariant. Applying variation of constants to \eqref{eq:contracting_error} with \eqref{eq:transition_function} establishes
\eqref{eq:robust_contraction_bound}.
\end{proof}
Uniform upper and lower bounds on the directional derivative of $e$ along normal segments imply the local comparison $\abs{e}=\Theta(\dist(\mathbf{x},\calM_{\mathrm{IT}}(t)))$ in the tube. Thus, \Cref{thm:contraction} relates timing-error decay to local distance from the moving time-to-go isochron. Perturbations tangent to the isochron remain governed by $a_{\mathrm{H}}$, domain constraints, or additional states and inputs.
\begin{remark}\label{rem:biased_pn}
If $a_{\mathrm{H}}=N V_{\mathrm{P}}\dot{\theta}$, then
\begin{align}
    a_{\mathrm{P}}^{\star}
    =&
    N V_{\mathrm{P}}\dot{\theta}
    -
    \frac{
    \varepsilon_{\calT}
    +\hat{\Delta}_{\calT}
    +k(t)e
    }{B_{\calT}},
    \label{eq:biased_pn_form}
\end{align}
which has the conventional biased-PN command structure. The scalar identity $\dot{e}=-k e$ also admits a nominal convergence analysis using the Lyapunov function $e^2/2$. The contraction formulation complements that analysis by relating $\calT$ to the executing homing flow and by providing a coordinate-invariant normal metric. The accompanying results address the robust error bound, bounded-input feasibility, first capture, and handover as $B_{\calT}$ vanishes.
\end{remark}
\begin{figure}[!ht]
    \centering
    \begin{tikzpicture}[
        font=\footnotesize,
        axis/.style={-{Latex[length=1.4mm,width=1mm]}, draw=designgray, line width=0.45pt},
        variation/.style={-{Latex[length=1.7mm,width=1.2mm]}, line width=0.9pt},
        path/.style={-{Latex[length=1.7mm,width=1.2mm]}, line width=0.9pt},
        level/.style={draw=designblue, line width=0.85pt}
    ]
        \foreach \offset/\shade in {-1.25/designblue,1.25/designorange} {
            \path[fill=\shade!8, draw=none]
                plot[domain=0.90:3.85, samples=60, variable=\y]
                    ({3.10+0.28*(\y-1.75)^2},\y)
                -- plot[domain=3.85:0.90, samples=60, variable=\y]
                    ({3.10+\offset+0.28*(\y-1.75)^2},\y)
                -- cycle;
        }
        \draw[axis] (0.55,0.65) -- (7.55,0.65) node[below] {$\xi_1$};
        \draw[axis] (0.55,0.65) -- (0.55,4.10) node[left] {$\xi_2$};
        \foreach \offset in {-1.25,1.25} {
            \draw[designgray!65, densely dashed, line width=0.55pt,
                domain=0.90:3.85, samples=60, variable=\y]
                plot ({3.10+\offset+0.28*(\y-1.75)^2},\y);
        }
        \draw[level, domain=0.90:3.85, samples=60, variable=\y]
            plot ({3.10+0.28*(\y-1.75)^2},\y);
        \node[above, text=designgray] at (3.0848,3.87) {$e<0$};
        \node[above, align=center, text=designblue] at (4.3348,3.87)
            {$\calM_{\mathrm{IT}}(t)$\\$e=0$};
        \node[above, text=designgray] at (5.5848,3.87) {$e>0$};
        \draw[designgray!65, densely dotted, line width=0.55pt]
            (3.10,3.05) -- (4.65,3.05) -- (4.65,1.75);
        \draw[variation, designgreen] (3.10,1.75) -- (3.10,3.05);
        \node[anchor=east, align=right, text=designgreen] at (2.93,2.50)
            {$\delta\mathbf{x}_{\parallel}$\\$\delta e=0$};
        \draw[variation, designorange] (3.10,1.75) -- (4.65,1.75);
        \node[below, text=designorange] at (3.96,1.67)
            {$\delta\mathbf{x}_{\perp}$};
        \draw[variation, black!85] (3.10,1.75)
            -- node[above, sloped, pos=0.61] {$\delta\mathbf{x}$} (4.65,3.05);
        \fill (3.10,1.75) circle (1.5pt);
        \node[below left] at (3.06,1.72) {$\mathbf{x}_a$};
        \fill (3.9092,3.45) circle (1.5pt);
        \node[below right] at (4.00,3.44) {$\mathbf{x}_b$};
        \fill (4.65,3.05) circle (1.3pt);
        \node[anchor=west, align=left] at (5.15,2.40)
            {$\delta e=\mathbf{p}\trans\delta\mathbf{x}_{\perp}$\\
             $\mathbf{p}=\nabla_{\mathbf{x}}e$};
        \node at (6.25,1.85)
            {$\delta\mathbf{x}=\delta\mathbf{x}_{\parallel}+\delta\mathbf{x}_{\perp}$};
        \node at (6.25,1.55)
            {$\delta\mathbf{x}\trans M_{\perp}\delta\mathbf{x}=(\delta e)^2$};

        \node at (4.05,0.2)
            {(a) Timing variations at a fixed time.};

        \begin{scope}[yshift=-3.60cm]
            \draw[axis] (0.75,0.45) -- (7.55,0.45)
                node[below] {$\calT_{\mathrm{H}}$};
            \draw[axis] (0.75,0.45) -- (0.75,3.25) node[left] {$\eta$};
            \fill[black!12] (0.61,0.72) rectangle (0.89,2.98);
            \draw[line width=1.1pt] (0.75,0.72) -- (0.75,2.98);
            \node[anchor=east] at (0.53,1.85) {$\calC$};
            \node[below] at (0.75,0.45) {$0$};
            \foreach \xx/\index in {6.15/0,4.35/1,2.55/2} {
                \draw[level, densely dashed] (\xx,0.72) -- (\xx,2.98);
                \node[above, text=designblue] at (\xx,3.02)
                    {$\calM_{\mathrm{IT}}(t_{\index})$};
                \node[below] at (\xx,0.45) {$t_f-t_{\index}$};
            }
            \draw[path, designgreen, domain=6.15:0.75, samples=80, variable=\x]
                plot (\x,{2.60-0.60*sin(180*(\x-0.75)/5.4)+0.15*(\x-0.75)/5.4});
            \draw[path, designorange, domain=6.15:0.75, samples=80, variable=\x]
                plot (\x,{0.90+0.45*(1-(\x-0.75)/5.4)^2+0.30*sin(180*(\x-0.75)/5.4)});
            \foreach \xx in {6.15,4.35,2.55} {
                \fill[designgreen] (\xx,{2.60-0.60*sin(180*(\xx-0.75)/5.4)+0.15*(\xx-0.75)/5.4}) circle (1.5pt);
                \fill[designorange] (\xx,{0.90+0.45*(1-(\xx-0.75)/5.4)^2+0.30*sin(180*(\xx-0.75)/5.4)}) circle (1.5pt);
            }
            \node[right, text=designgreen] at (6.22,2.75) {$\mathbf{x}_a(t_0)$};
            \node[right, text=designorange] at (6.22,0.90) {$\mathbf{x}_b(t_0)$};
            \node[right] at (0.97,2.75) {$t_f$};
            \node[right] at (0.97,1.13) {$t_f$};
            \node at (4.05,-0.25)
                {(b) A common arrival time along different paths.};
        \end{scope}
    \end{tikzpicture}
    \caption{Schematic state-space views of the scheduled isochron: (a) tangent and normal perturbations in local state coordinates $(\xi_1,\xi_2)$; (b) nominal homing with an exact time-to-go coordinate.}
    \label{fig:isochron_geometry}
\end{figure}

In \Cref{fig:isochron_geometry}(a), the shaded bands indicate negative and positive timing-error regions. The infinitesimal perturbation $\delta\mathbf{x}$ has tangent and normal components. Since $\mathbf{p}\trans\delta\mathbf{x}_{\parallel}=0$, only $\delta\mathbf{x}_{\perp}$ contributes to $\delta e$ and the metric value $(\delta e)^2$. Panel (b) uses the exact time-to-go $\calT_{\mathrm{H}}$ and an internal coordinate $\eta$. Under nominal baseline homing, \eqref{eq:transport_pde} implies $\dot{\calT}_{\mathrm{H}}=-1$, so the illustrated trajectories pass through successive scheduled sets at $t_0<t_1<t_2<t_f$. The distinct paths share the arrival time $t_f$ by \eqref{eq:baseline_hitting_time}.

We next incorporate the acceleration limit through a timing funnel. Let $\rho$ be continuously differentiable on $[t_0,t_f)$, with $0<\rho(t)<s(t)$ and $\lim_{t\to t_f^-}\rho(t)=0$. Define
$\calF(t)=\{\mathbf{x}:\abs{e(\mathbf{x},t)}\leq\rho(t)\}$ and let $\ell
    \coloneqq
    F_{\calT}+\hat{\Delta}_{\calT}$ with $\abs{\tilde{\Delta}_{\calT}}
    \leq
    \bar{\delta}$. Among the commands satisfying the robust funnel inequalities, select the command closest to the baseline homing law:
    \begin{subequations}\label{eq:optimization}
        \begin{align}
    a_{\mathrm{cf}}
    \in&
    \arg\min_{\abs{a}\leq a_{\max}}
    \frac{1}{2}(a-a_{\mathrm{H}})^2
    \label{eq:filter_objective}
        \end{align}
The constraints account for the residual bound at the upper and lower funnel boundaries:
        \begin{align}
    B_{\calT}a
    \leq&
    -\kappa(e-\rho)-\ell-\bar{\delta}+\dot{\rho},
    \label{eq:filter_upper}
    \\
    B_{\calT}a
    \geq&
    -\kappa(e+\rho)-\ell+\bar{\delta}-\dot{\rho},
    \label{eq:filter_lower}
\end{align}
 \end{subequations}
where $\kappa(t)>0$. The optimization \eqref{eq:optimization} projects the baseline command onto the feasible acceleration interval and thereby defines a scalar command filter. Feasibility of the constraint interval determines whether the prescribed funnel is compatible with the available lateral acceleration.
\begin{lemma}
\label{prop:feasibility}
Let the lower and upper bounds in \eqref{eq:filter_lower} and \eqref{eq:filter_upper} be denoted by
\begin{align}
    \zeta_-
    \coloneqq&
    -\kappa(e+\rho)-\ell+\bar{\delta}-\dot{\rho},
    \nonumber\\
    \zeta_+
    \coloneqq&
    -\kappa(e-\rho)-\ell-\bar{\delta}+\dot{\rho}.
    \label{eq:zeta_bounds}
\end{align}
The filter \eqref{eq:filter_objective}--\eqref{eq:filter_lower}
is feasible if and only if
\begin{align}
    [\zeta_-,\zeta_+]
    \cap
    [-\abs{B_{\calT}}a_{\max},
      \abs{B_{\calT}}a_{\max}]
    \neq&
    \varnothing .
    \label{eq:interval_intersection}
\end{align}
Equivalently, $\zeta_-\leq\zeta_+,~
    \zeta_-\leq\abs{B_{\calT}}a_{\max},~
    -\abs{B_{\calT}}a_{\max}\leq\zeta_+$. If $\zeta_-\leq\zeta_+$ and $B_{\calT}\neq0$, the minimum required symmetric authority is
\begin{align}
    a_{\mathrm{req}}
    =&
    \frac{
    \max\{\zeta_-,-\zeta_+,0\}
    }{\abs{B_{\calT}}}.
    \label{eq:required_acceleration}
\end{align}
\end{lemma}
\begin{proof}
Fix $(\mathbf{x},t)$ and set $b=\abs{B_{\calT}}a_{\max}$. Under the input bound in \eqref{eq:filter_objective}, the scalar quantity $B_{\calT}a$ ranges over $[-b,b]$. The constraints \eqref{eq:filter_upper}--\eqref{eq:filter_lower} and the definitions \eqref{eq:zeta_bounds} require $\zeta_-\leq B_{\calT}a\leq\zeta_+$. Thus, a feasible command exists if and only if the attainable and required intervals intersect, establishing \eqref{eq:interval_intersection}. Comparing the endpoints, the overlap condition takes the form $\max\{\zeta_-,-b\}\leq\min\{\zeta_+,b\}$. Since $b\geq0$, the endpoint inequality is equivalent to $\zeta_-\leq\zeta_+$, $\zeta_-\leq b$, and $-\zeta_+\leq b$. For a nonempty constraint interval, the smallest admissible value of $b$ is $\max\{\zeta_-,-\zeta_+,0\}$. Dividing by $\abs{B_{\calT}}$ when $B_{\calT}\neq0$ yields the minimum required acceleration in \eqref{eq:required_acceleration}.
\end{proof}
When $B_{\calT}=0$, lateral acceleration has zero first-order timing effectiveness irrespective of the feedback gain; feasibility then requires $0\in[\zeta_-,\zeta_+]$. The interval width obtained from \eqref{eq:zeta_bounds} is
$\zeta_+-\zeta_-=2(\kappa\rho+\dot{\rho}-\bar{\delta})$. 
Thus, excessive uncertainty or rapid funnel shrinkage can make the interval empty even for large $a_{\max}$.
\begin{corollary}
\label{cor:funnel_compatibility}
A necessary condition for robust funnel feasibility is
\begin{align}
\kappa(t)\rho(t)+\dot{\rho}(t)\geq&\ \bar{\delta}(t).
\label{eq:intrinsic_funnel_condition}
\end{align}
Condition \eqref{eq:intrinsic_funnel_condition} is independent of the available acceleration authority. Once this condition holds, the remaining limitation is expressed by the steering margin
\begin{align}
\mathfrak{m}(t)\coloneqq&\ |B_{\calT}|a_{\max}
-\max\{\zeta_-,-\zeta_+,0\}.
\label{eq:authority_margin}
\end{align}
The inequalities $\mathfrak{m}(t)\geq0$ and $\zeta_-\leq\zeta_+$ are equivalent to pointwise feasibility when $B_{\calT}\neq0$. For the power-law funnel $\rho=\rho_0[(t_h-t)/(t_h-t_0)]^p$ and $\kappa=c/(t_h-t)$, the identity $\kappa\rho+\dot{\rho}=(c-p)\rho/(t_h-t)$ makes $c\geq p$ necessary in the nominal case, while persistent uncertainty imposes the stronger lower bound in \eqref{eq:intrinsic_funnel_condition}.
\end{corollary}
With vector actuation, the timing channel remains scalar. The minimum-norm timing correction acts along the input sensitivity, and a Euclidean input bound determines the attainable timing action.
\begin{corollary}
\label{cor:vector_input}
Consider $\dot{\mathbf{x}}=\mathbf{f}(\mathbf{x},t)+G(\mathbf{x},t)\mathbf{u}+\mathbf{d}(\mathbf{x},t)$ with $\mathbf{u}\in\R^m$, $\|\mathbf{u}\|_2\leq u_{\max}$, and a baseline command $\mathbf{u}_{\mathrm{H}}$. Define the predictor defect along the baseline vector flow by
$\varepsilon_{\calT}=\partial_t\calT+\nabla_{\mathbf{x}}\calT\trans(\mathbf{f}+G\mathbf{u}_{\mathrm{H}})+1$, and define the row vector $B_{\calT}\coloneqq\nabla_{\mathbf{x}}\calT\trans G$.
Where $B_{\calT}\neq0$, the minimum-norm timing correction
\begin{align}
\mathbf{u}_{\mathrm{IT}}^{\star}
=&-\frac{B_{\calT}\trans}{\|B_{\calT}\|_2^2}
\left(\varepsilon_{\calT}+\hat{\Delta}_{\calT}+k(t)e\right)
\label{eq:vector_timing_correction}
\end{align}
recovers the timing dynamics \eqref{eq:contracting_error}, so \Cref{thm:contraction} applies. The attainable scalar timing action satisfies $B_{\calT}\mathbf{u}\in[-\|B_{\calT}\|_2u_{\max},\|B_{\calT}\|_2u_{\max}]$.
The feasibility criteria in \Cref{prop:feasibility,cor:funnel_compatibility} extend by replacing $|B_{\calT}|a_{\max}$ with $\|B_{\calT}\|_2u_{\max}$. The extension includes three-dimensional point-mass guidance with two lateral steering components. The reader may refer to \cite{sinha2021wide,sinha2021event} and references therein for details on impact time guidance in three-dimensional settings.
\end{corollary}
Terminal calibration connects the feasible timing funnel to physical capture, as discussed next.
\begin{theorem}
\label{thm:first_hitting}
Suppose \Cref{ass:coordinate} holds,
$\mathbf{x}(t_0)\in\calF(t_0)$, and
\eqref{eq:interval_intersection} holds for every state visited in
$\calF(t)$, $t<t_f$. Then
\eqref{eq:optimization}
makes $\calF(t)$ robustly forward invariant for every mismatch satisfying
$\abs{\tilde{\Delta}_{\calT}}\leq\bar{\delta}$. The first capture occurs at the prescribed absolute time $t_f$.
\end{theorem}
\begin{proof}
By \Cref{prop:feasibility}, the feasibility condition \eqref{eq:interval_intersection} ensures that the optimization problem \eqref{eq:optimization} admits a command $a_{\mathrm{cf}}$ satisfying $\abs{a_{\mathrm{cf}}}\leq a_{\max}$. Set $g_+=e-\rho$. Applying the upper constraint \eqref{eq:filter_upper} to the timing channel \eqref{eq:timing_channel_total} yields
$\dot{g}_+\leq-\kappa g_+$ under the worst positive mismatch. Writing the lower boundary as $g_-=-e-\rho$, the constraint \eqref{eq:filter_lower} implies
$\dot{g}_-\leq-\kappa g_-$ under the worst negative mismatch. The initial funnel membership implies $g_+(t_0),g_-(t_0)\leq0$. Scalar comparison preserves both inequalities and establishes funnel invariance. For
$t<t_f$, the error definition \eqref{eq:timing_error} and the funnel bound imply $\calT(\mathbf{x}(t),t)
    =s(t)+e(t)
    \geq s(t)-\rho(t)>0$. Since $\calT=0$ on $\calC$ by \Cref{ass:coordinate}, capture cannot occur before $t_f$. As $t\to t_f^-$,
$\calT\to0$. The sublevel condition \eqref{eq:terminal_sublevel} places the trajectory in $\mathcal{N}_{\calC}$ for all sufficiently late times, and terminal calibration in \eqref{eq:terminal_calibration} implies $q(\mathbf{x}(t))\to0$. The continuous extension in \Cref{ass:coordinate} and closedness of $\calC$ imply
$\mathbf{x}(t_f)\in\calC$.
\end{proof}
The first-capture guarantee requires a coordinate satisfying \Cref{ass:coordinate}. We use the transport equation \eqref{eq:transport_pde} to construct an exact PN coordinate and examine alternative predictors under the same actuator model. Consider a stationary target,
$V_{\mathrm{P}}=V$, and the velocity-normal PN baseline
\begin{align}
    a_{\mathrm{H}}
    =&
    N V_{\mathrm{P}}\dot{\theta}
    =
    -\frac{NV^2}{r}\sin{\sigma}.
    \label{eq:pn_baseline}
\end{align}
We seek an exact homogeneous time-to-go map of the form
\begin{align}
    \calT_{\mathrm{H}}(r,\sigma)
    =&
    \frac{r}{V}h_N(\sigma).
    \label{eq:exact_pn_tgo}
\end{align}
Substitution into \eqref{eq:transport_pde} yields the scalar equation
\begin{align}
    (N-1)\sin{\sigma}\,h_N'(\sigma)
    +\cos{\sigma}\,h_N(\sigma)
    =&
    1;~
    h_N(0)=1.
    \label{eq:h_ode}
\end{align}
For $N>1$, \eqref{eq:h_ode} can be tabulated on any compact capturable subset
$\abs{\sigma}\leq\sigma_{\max}<\pi$.
The lateral timing coefficient of $\calT_{\mathrm{H}}$ is
\begin{align}
    B_{\mathrm{H}}
    =&
    \frac{r}{V^2}h_N'(\sigma).
    \label{eq:exact_pn_B}
\end{align}
The regular expansion
\begin{align}
    h_N(\sigma)
    =&
    1+\frac{\sigma^2}{4N-2}
    +\mathcal{O}(\sigma^4)
    \label{eq:h_expansion}
\end{align}
recovers the leading-order behavior of the large-heading approximation with
$\sin^2{\sigma}/(4N-2)$ used in the guidance design of \cite{sinha2020super}. The large-heading approximation differs from the exact solution of
\eqref{eq:h_ode}.
\begin{table*}[!ht]
\caption{Time-to-go coordinates and the corresponding coefficients of velocity-normal acceleration in $\dot{e}=F_i+B_i a_{\mathrm{P}}$.}
\centering
\footnotesize
\setlength{\tabcolsep}{3pt}
\renewcommand{\arraystretch}{1.08}
\begin{tabular}{@{}p{0.475\linewidth}
    >{\centering\arraybackslash}p{0.475\linewidth}@{}}
\toprule
\multicolumn{1}{c}{\textbf{Time-to-go map} $\calT_i$}
&
\multicolumn{1}{c}{\textbf{Coefficient} $B_i$ \textbf{of} $a_{\mathrm{P}}$ \textbf{in} $\dot{e}$}
\\
\midrule

Conservative, based on closing speed:~~
$\displaystyle
\calT_{\mathrm{C}}=-\frac{r}{V_r}$
&
$\displaystyle
B_{\mathrm{C}}=\frac{r\sin{\sigma}}{V_r^2}$
\\[2pt]

Deviated pursuit \cite{sinha2025deviated}:~~
$\displaystyle
\calT_{\mathrm{DP}}
=
\frac{r\sec{\sigma}}{\Upsilon}
\left[
V_{\mathrm{P}}
+V_{\mathrm{T}}\cos{\left(\psi_{\mathrm{T}}+\sigma\right)}
\right]$
&
$\displaystyle
B_{\mathrm{DP}}
=
-\frac{rV_{\theta}\sec^2{\sigma}}
{V_{\mathrm{P}}\Upsilon}$
\\[2pt]

TPN-inspired \cite{kumar2022true}:~~
$\displaystyle
\calT_{\mathrm{TPN}}
=
-\frac{r\left(V_r+2c_{\mathrm{T}}\right)}
{V_{\theta}^2+V_r^2+2c_{\mathrm{T}}V_r}$
&
$\displaystyle
\begin{aligned}
B_{\mathrm{TPN}}
=&
\frac{r\Big[
\sin{\sigma}
\big(
\left(V_r+2c_{\mathrm{T}}\right)^2-V_{\theta}^2
\big)-2\left(V_r+2c_{\mathrm{T}}\right)
V_{\theta}\cos{\sigma}
\Big]}{
\left(
V_{\theta}^2+V_r^2+2c_{\mathrm{T}}V_r
\right)^2}
\end{aligned}$
\\[2pt]

Large-heading PN
\cite{kumar2015large,sinha2020super}:~~
$\displaystyle
\calT_{\mathrm{LPN}}
=
\frac{r}{V_{\mathrm{P}}}
\left(
1+\frac{\sin^2{\sigma}}{4N-2}
\right)$
&
$\displaystyle
B_{\mathrm{LPN}}
=
\frac{r\sin{\left(2\sigma\right)}}
{V_{\mathrm{P}}^2\left(4N-2\right)}$
\\

\bottomrule
\end{tabular}
\label{tab:tgo_atlas}
\end{table*}

The PN geometry also determines the rate at which lateral timing authority vanishes near capture.
\begin{lemma}
\label{prop:pn_authority_order}
For a nontrivial stationary-target PN trajectory with $N>1$ that approaches collision-course capture,
\begin{align}
\sin{\sigma}=&\ C r^{N-1},
\label{eq:pn_sigma_range}
\end{align}
for a trajectory-dependent constant $C$. The exact PN timing coefficient in \eqref{eq:exact_pn_B} satisfies
\begin{align}
|B_{\mathrm{H}}|=&\ \Theta(r^N)=\Theta(\calT_{\mathrm{H}}^N).
\label{eq:pn_B_order}
\end{align}
If direct prescribed-time inversion is continued to capture with $k=c/s$, $e=\mathcal{O}(s^c)$, and vanishing predictor residual, then
\begin{align}
a_{\mathrm{IT}}=&\ \mathcal{O}(s^{c-N-1}).
\label{eq:pn_command_order}
\end{align}
Thus, $c\geq N+1$ is required for bounded nominal timing correction and $c>N+1$ makes the timing correction vanish at capture.
\end{lemma}
\begin{proof}
For the stationary-target PN baseline in \eqref{eq:pn_baseline}, the polar dynamics \eqref{eq:polar_1} and \eqref{eq:polar_4} reduce to $\dot{r}=-V\cos{\sigma}$ and $\dot{\sigma}=-(N-1)V\sin{\sigma}/r$. Hence, $d\sigma/dr=(N-1)\tan{\sigma}/r$, and integration establishes \eqref{eq:pn_sigma_range}. The expansion \eqref{eq:h_expansion} implies $h_N'(\sigma)=\sigma/(2N-1)+\mathcal{O}(\sigma^3)$. Substituting the derivative expansion into \eqref{eq:exact_pn_B} and using $\sigma=\Theta(r^{N-1})$ from \eqref{eq:pn_sigma_range}, one obtains $B_{\mathrm{H}}=\Theta(r^N)$. Since \eqref{eq:exact_pn_tgo} has the terminal form $\calT_{\mathrm{H}}=(r/V)(1+o(1))$, the two orders in \eqref{eq:pn_B_order} agree. Finally, substituting \eqref{eq:pn_B_order} and the assumed timing-error order into the nominal bias in \eqref{eq:ideal_command} with $k=c/s$ produces \eqref{eq:pn_command_order}.
\end{proof}
\Cref{tab:tgo_atlas} collects established time-to-go predictors and their input coefficients under the plant \eqref{eq:interceptor_kinematics}, where we write
$\psi_{\mathrm{T}}=\gamma_{\mathrm{T}}-\theta$ and
$\Upsilon=V_{\mathrm{P}}^2-V_{\mathrm{T}}^2$ for brevity. Each coefficient uses the velocity-normal actuation model, irrespective of the actuation assumed in the original derivation. For the closing-speed predictor, differentiating
$\calT_{\mathrm{C}}$ along \eqref{eq:polar_1}--\eqref{eq:polar_2} yields
\begin{align}
    \dot{e}
    =&
    \frac{V_{\theta}^2+r a_{\mathrm{T}r}}{V_r^2}
    +\frac{r\sin{\sigma}}{V_r^2}a_{\mathrm{P}}.
    \label{eq:conservative_channel}
\end{align}
The coefficient $B_{\mathrm{TPN}}$ likewise includes the radial and transverse sensitivities induced by a single velocity-normal command. A coefficient derived for LOS-normal actuation would correspond to a different physical input.

Each coordinate in \Cref{tab:tgo_atlas} has the factorized form
$r h(\cdot)$ with positive, bounded $h$ on a compact subset of the coordinate's domain, providing terminal calibration. Homing-flow compatibility is assessed by the defect \eqref{eq:predictor_defect}; the coefficient $B_i$ measures local timing authority; and \eqref{eq:interval_intersection} tests whether the available authority satisfies the filter constraints. The impact-time analysis accounts for all three properties. For example,
$B_{\mathrm{C}}=0$ at $\sigma=0$;
$B_{\mathrm{DP}}=0$ at $V_{\theta}=0$;
and $B_{\mathrm{LPN}}=0$ when
$\sin{(2\sigma)}=0$. Each surface corresponds to a loss of first-order timing authority and enters the nonsingularity analysis.

When the target has a constant velocity, the predicted interception point $\mathbf{p}_{\mathrm{I}}(t)
    =
    \mathbf{p}_{\mathrm{T}}(t)
    +\mathbf{v}_{\mathrm{T}}(t)(t_f-t)$ is stationary for a fixed $t_f$ and may replace the target in $r$, $\theta$, and $\sigma$. For a maneuvering target, one may instead use
$\hat{\mathbf{p}}_{\mathrm{T}}(t_f\mid t)$. The motion and error of the predicted interception point enter
$\varepsilon_{\calT}$ or $\Delta_{\calT}$. A prescribed-interception guarantee requires a target model, sufficient control authority, and terminally compatible prediction error. Arbitrary unknown target motion lies outside the scope of the present results.
\begin{remark}
\label{rem:autopilot}
The kinematics \eqref{eq:interceptor_kinematics} take achieved lateral acceleration as the input. A bounded static tracking error can be included in
$\tilde{\Delta}_{\calT}$. A first-order acceleration or turn-rate actuator makes the command-to-$e$ channel relative degree two and requires a dynamic contraction and filter design. Appending a lag to
\eqref{eq:ideal_command} changes the timing channel \eqref{eq:timing_channel_relative} and requires a new analysis. Range, LOS rate, the interceptor's heading, and the target's variables required by the selected map are presumed available, with measurement and estimation errors covered by the residual bound $\abs{\tilde{\Delta}_{\calT}}\leq\bar{\delta}$. Actuator and sensing dynamics remain topics for further study.
\end{remark}

If direct inversion is continued to $t_f$ and
$\abs{B_{\calT}}=\Theta(s^{\mu})$, a sufficient bounded-command condition is
$e=\mathcal{O}(s^{\mu+1})$ and
$\varepsilon_{\calT}+\tilde{\Delta}_{\calT}
=\mathcal{O}(s^{\mu})$. With $k=c/s$, the stated decay orders are ensured in the worst case by $c>\mu+1$ and a residual of order at least $s^{\mu}$. When $\mu>0$, a bounded residual can still lead to an unbounded acceleration command as timing authority vanishes.
\begin{remark}
The invariant funnel in \Cref{thm:first_hitting} bounds the timing error while the feasibility condition \eqref{eq:interval_intersection} holds. Near capture, $B_{\calT}\to0$ for many lateral-guidance maps as collision-course alignment develops. A staged architecture corrects the schedule while lateral steering retains timing authority, aligns the time-to-go coordinate at a preterminal instant, and returns to ordinary homing for the terminal phase.
\end{remark}
\begin{theorem}
\label{thm:handover}
Suppose $\calT=\calT_{\mathrm{H}}$ is the exact baseline time-to-go coordinate defined by \eqref{eq:baseline_hitting_time} and satisfies \eqref{eq:transport_pde}. Choose
$t_0<t_h<t_f$ and assume
$\abs{B_{\calT}}\geq b_h>0$ on $[t_0,t_h]$.
For $c>1$, apply
\begin{equation}
    a_{\mathrm{P}}
    =
    \begin{cases}
        \displaystyle
        a_{\mathrm{H}}
        -\frac{c\,e}{(t_h-t)B_{\calT}},
        & t<t_h,\\[2mm]
        a_{\mathrm{H}},
        & t\geq t_h.
    \end{cases}
    \label{eq:handover_guidance}
\end{equation}
If the command \eqref{eq:handover_guidance} satisfies $\abs{a_{\mathrm{P}}}\leq a_{\max}$ and
$\abs{e(t)}<s(t)$ during alignment, then
\begin{align}
    e(t)
    =&
    e(t_0)
    \left(
    \frac{t_h-t}{t_h-t_0}
    \right)^c,
    ~~ t<t_h,
    \label{eq:handover_error}
\end{align}
so $e(t)=0$ for $t\in[t_h,t_f]$, and the first nominal capture occurs at $t_f$. The timing correction is bounded and tends to zero at $t_h$.
\end{theorem}
\begin{proof}
Before handover, substituting \eqref{eq:transport_pde} and
\eqref{eq:handover_guidance} into the nominal timing channel \eqref{eq:timing_channel_relative} reduces the error dynamics to
$\dot{e}=-c e/(t_h-t)$, which integrates to
\eqref{eq:handover_error}. For the correction in \eqref{eq:handover_guidance}, the response \eqref{eq:handover_error} and the assumed bound $\abs{B_{\calT}}\geq b_h$ imply
\begin{align}
    \abs{a_{\mathrm{P}}-a_{\mathrm{H}}}
    \leq&
    \frac{c\abs{e(t_0)}}{b_h(t_h-t_0)^c}
    (t_h-t)^{c-1},
\end{align}
so the timing correction vanishes for $c>1$. At handover, \eqref{eq:handover_error} and \eqref{eq:timing_error} imply $\calT_{\mathrm{H}}(\mathbf{x}(t_h),t_h)=t_f-t_h$. The baseline first-hitting-time definition \eqref{eq:baseline_hitting_time} then places the first post-handover capture at $t_f$, while \eqref{eq:transport_pde} preserves $e=0$. During alignment, $\abs{e(t)}<s(t)$ ensures $\calT_{\mathrm{H}}>0$ by \eqref{eq:timing_error}, excluding earlier capture.
\end{proof}
The handover in \Cref{thm:handover} requires alignment with the scheduled isochron of the executing baseline homing flow. For an approximate predictor, a small numerical value of $e$ must be accompanied by a compatibility check through \eqref{eq:predictor_defect}: a nonzero $\varepsilon_{\calT}+\Delta_{\calT}$ after $t_h$ causes drift from the scheduled level set by \eqref{eq:timing_channel_relative}. An exact numerical time-to-go map, continued use of the filter, or a robust timing tube addresses the corresponding prediction mismatch.
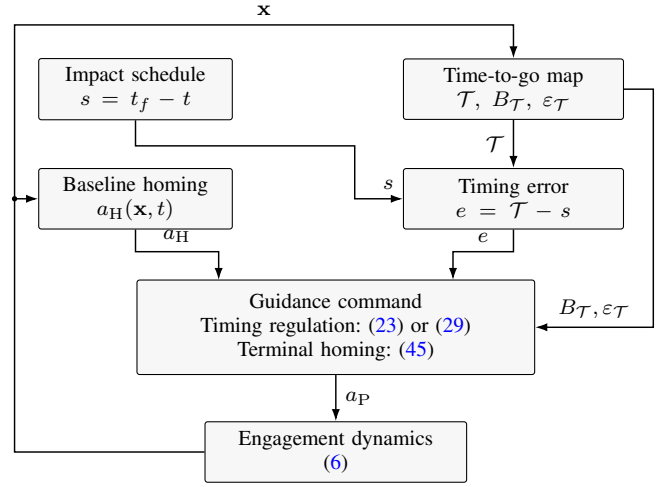
\begin{figure}[!ht]
    \centering
    \begin{tikzpicture}[
        font=\footnotesize,
        block/.style={draw, rounded corners=1pt, align=center,
            inner sep=3pt, minimum height=0.8cm, fill=black!3},
        signal/.style={-{Latex[length=1.6mm,width=1.1mm]}, line width=0.55pt}
    ]
        \node[block, text width=2.35cm] (schedule) at (1.5,0)
            {Impact schedule\\$s=t_f-t$};
        \node[block, text width=2.7cm] (map) at (6.5,0)
            {Time-to-go map\\$\calT,~B_{\calT},~\varepsilon_{\calT}$};
        \node[block, text width=2.35cm] (homing) at (1.5,-1.45)
            {Baseline homing\\$a_{\mathrm{H}}(\mathbf{x},t)$};
        \node[block, text width=2.7cm] (error) at (6.5,-1.45)
            {Timing error\\$e=\calT-s$};
        \node[block, text width=5.05cm, minimum height=1.25cm]
            (command) at (4.15,-3.15)
            {Guidance command\\
             Timing regulation: \eqref{eq:ideal_command} or \eqref{eq:optimization}\\
             Terminal homing: \eqref{eq:handover_guidance}};
        \node[block, text width=3.25cm] (plant) at (4.15,-4.8)
            {Engagement dynamics\\\eqref{eq:affine_plant}};

        \draw[signal] (schedule.south) -- (1.5,-0.75)
            -- (4.4,-0.75) |- node[pos=0.85,above] {$s$} (error.west);
        \draw[signal] (map.south) -- node[left] {$\calT$} (error.north);
        \draw[signal] (homing.south) -- (1.5,-2.15)
            -- node[above] {$a_{\mathrm{H}}$} (2.6,-2.15)
            -- (2.6,-2.525);
        \draw[signal] (error.south) -- (6.5,-2.15)
            -- node[above] {$e$} (5.7,-2.15)
            -- (5.7,-2.525);
        \draw[signal] (map.east) -- (8.35,0) -- (8.35,-3.15)
            -- node[above] {$B_{\calT},\varepsilon_{\calT}$} (command.east);
        \draw[signal] (command.south) -- node[right] {$a_{\mathrm{P}}$} (plant.north);
        \draw[signal] (plant.west) -- (-0.1,-4.8) -- (-0.1,0.85)
            -- node[above] {$\mathbf{x}$} (6.5,0.85) -- (map.north);
        \fill (-0.1,-1.45) circle (1pt);
        \draw[signal] (-0.1,-1.45) -- (homing.west);
    \end{tikzpicture}
    \caption{Impact-time guidance and homing handover.}
    \label{fig:guidance_architecture}
\end{figure}
In \Cref{fig:guidance_architecture}, state feedback updates the baseline homing command and the time-to-go map. The error $e=\calT-(t_f-t)$ measures deviation from the assigned schedule, while $B_{\calT}$ and $\varepsilon_{\calT}$ quantify lateral timing sensitivity and mismatch with the baseline flow. Inversion or constrained projection determines $a_{\mathrm{P}}$, which drives the engagement dynamics. Under \Cref{thm:handover}, nominal alignment at $t_h$ allows baseline homing to complete interception at $t_f$.

\section{Simulations}\label{sec:simulations}
The numerical study examines timing regulation for stationary and maneuvering targets. Two cases distinguish regulation using the exact homing-compatible coordinate from evaluation of closed-form predictors: a stationary-target case with the large-heading PN (LPN) approximation and a maneuvering-target case with the deviated-pursuit predictor. The reported histories include trajectories, relative range, time-to-go, timing error, lateral acceleration, timing effectiveness, and predictor defect. The preterminal handover time $t_h$ marks the transition from timing correction to baseline homing. The steering margin $\mathfrak{m}$ in \eqref{eq:authority_margin} provides the online feasibility diagnostic described in \Cref{cor:funnel_compatibility}.

\begin{figure*}[h!]
	\centering
	\captionsetup[subfigure]{font=footnotesize,skip=1pt}
	
	\begin{subfigure}[t]{0.43\textwidth}
		\centering
		\includegraphics[width=\linewidth]{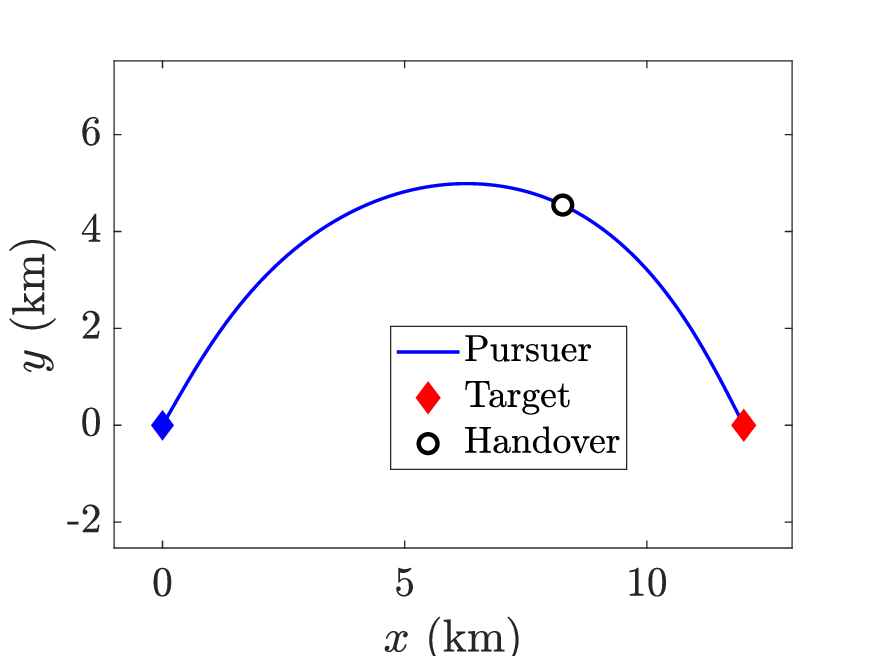}
		\caption{Interceptor's trajectory.}
		\label{fig:lpn_traj}
	\end{subfigure}
	\hspace{0.04\textwidth}
	\begin{subfigure}[t]{0.43\textwidth}
		\centering
		\includegraphics[width=\linewidth]{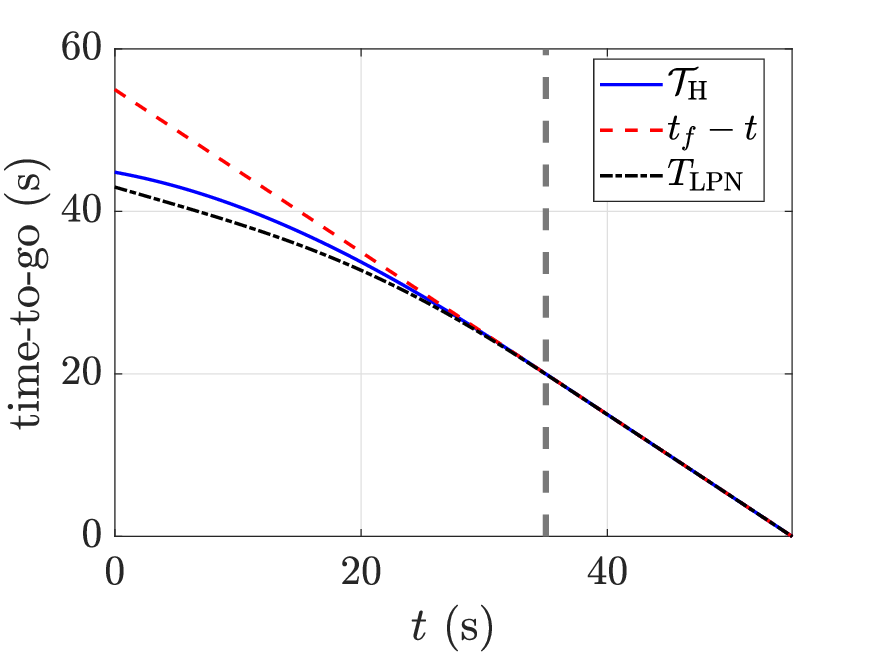}
		\caption{Time-to-go profiles.}
		\label{fig:lpn_tgo}
	\end{subfigure}
	
	\begin{subfigure}[t]{0.43\textwidth}
		\centering
		\includegraphics[width=\linewidth]{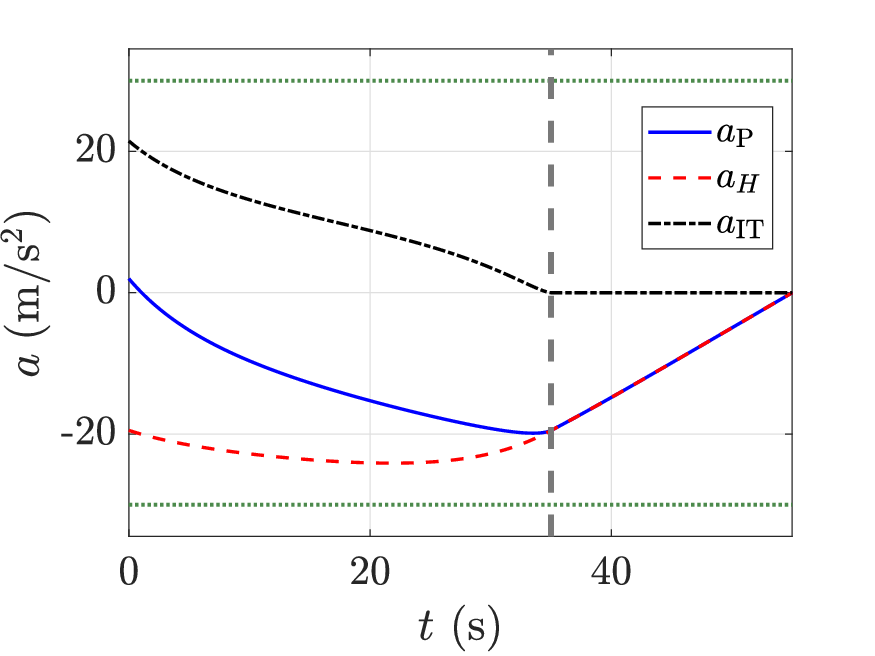}
		\caption{Lateral acceleration profiles.}
		\label{fig:lpn_aP}
	\end{subfigure}
	\hspace{0.04\textwidth}
	\begin{subfigure}[t]{0.43\textwidth}
		\centering
		\includegraphics[width=\linewidth]{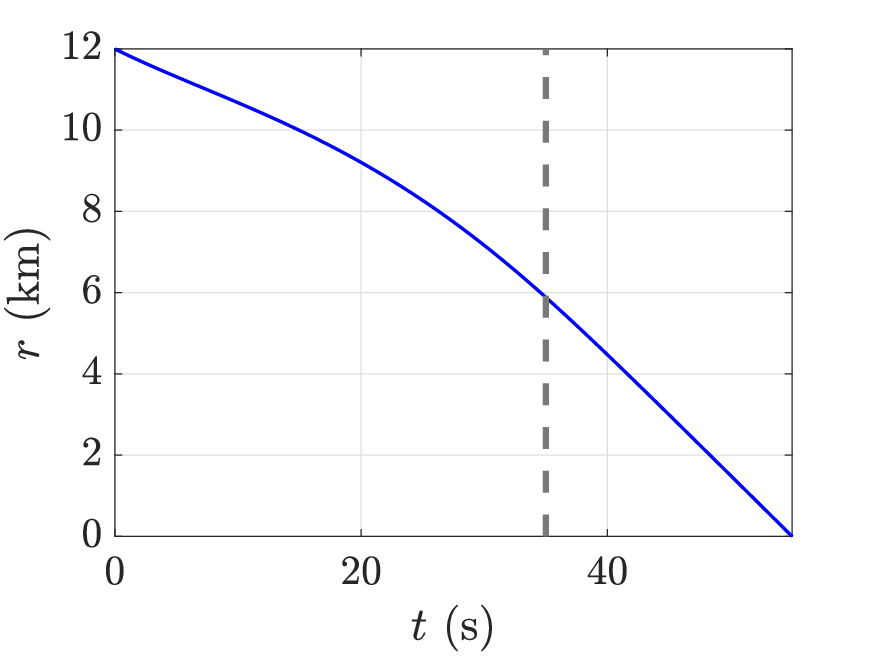}
		\caption{Relative range profile.}
		\label{fig:lpn_range}
	\end{subfigure}
	
	\begin{subfigure}[t]{0.43\textwidth}
		\centering
		\includegraphics[width=\linewidth]{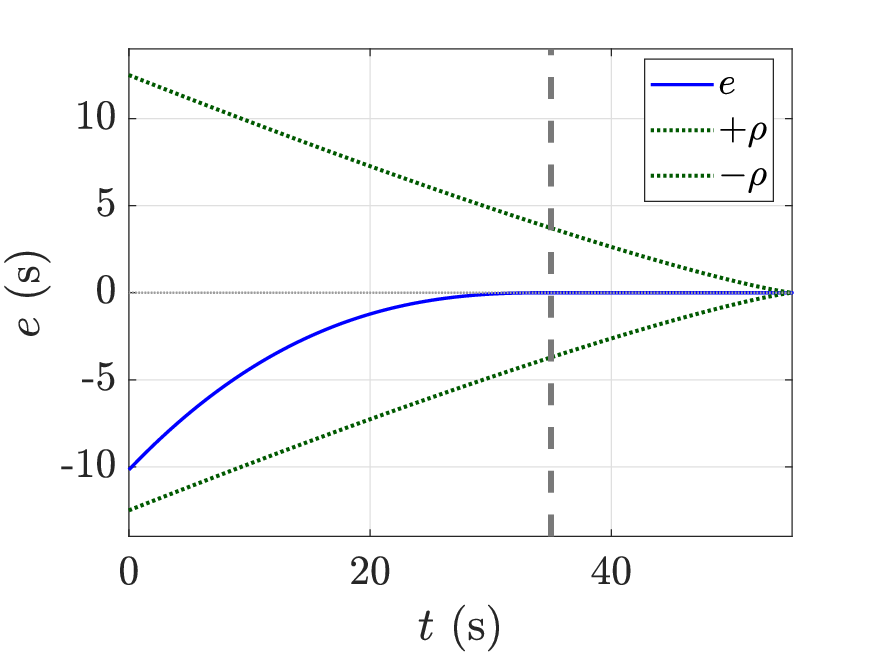}
		\caption{Timing error and prescribed funnel.}
		\label{fig:lpn_errors}
	\end{subfigure}
	\hspace{0.04\textwidth}
	\begin{subfigure}[t]{0.43\textwidth}
		\centering
		\includegraphics[width=\linewidth]{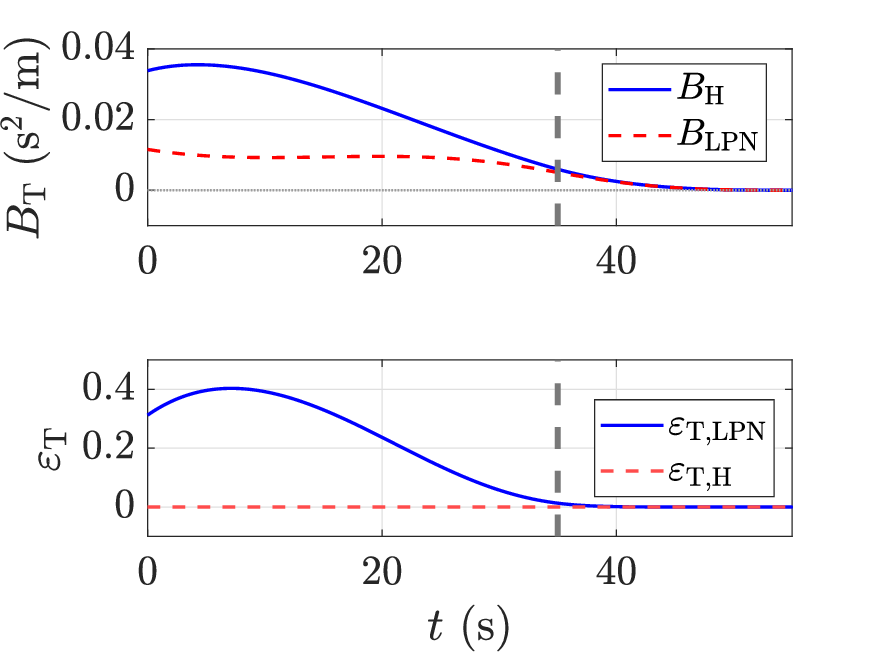}
		\caption{Timing effectiveness and predictor defect.}
		\label{fig:lpn_bT}
	\end{subfigure}
	
	\caption{Performance under large-heading PN-based impact-time guidance.}
	\label{fig:lpn_sim}
\end{figure*}

\begin{figure*}[h!]
	\centering
	\captionsetup[subfigure]{font=footnotesize,skip=1pt}
	
	\begin{subfigure}[t]{0.43\textwidth}
		\centering
		\includegraphics[width=\linewidth]{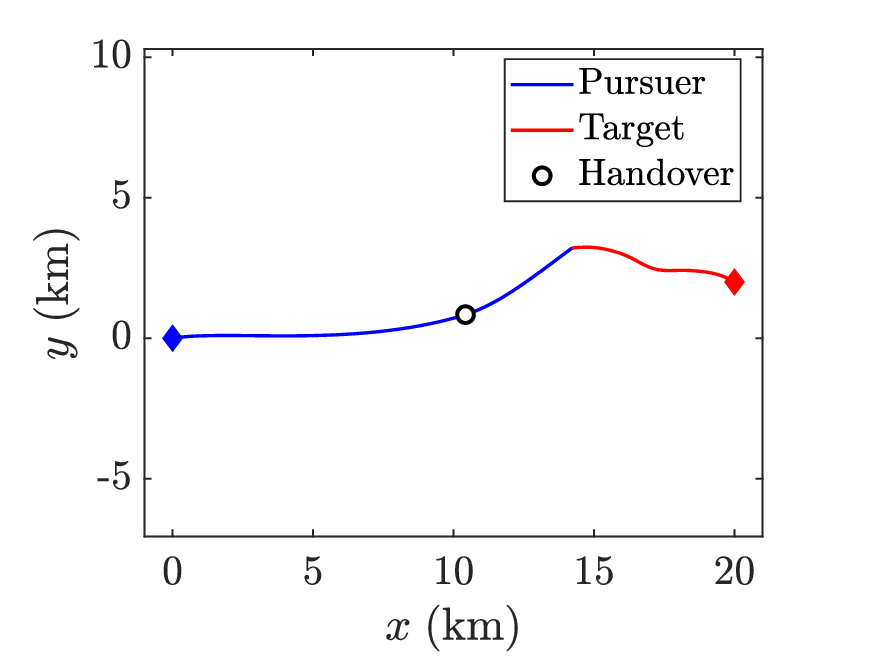}
		\caption{Agents' trajectories.}
		\label{fig:devia_traj}
	\end{subfigure}
	\hspace{0.04\textwidth}
	\begin{subfigure}[t]{0.43\textwidth}
		\centering
		\includegraphics[width=\linewidth]{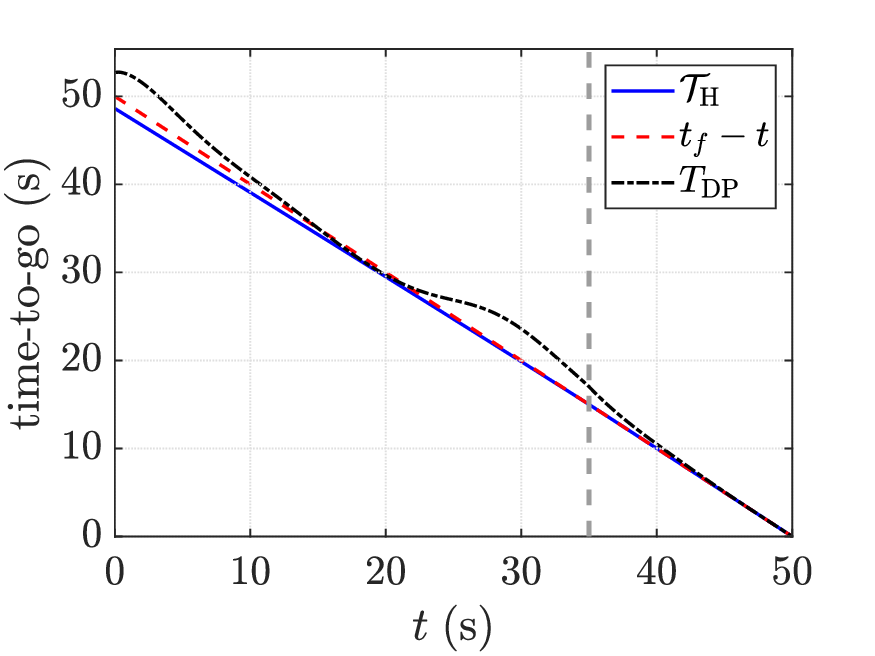}
		\caption{Time-to-go profiles.}
		\label{fig:devia_tgo}
	\end{subfigure}
	
	\begin{subfigure}[t]{0.43\textwidth}
		\centering
		\includegraphics[width=\linewidth]{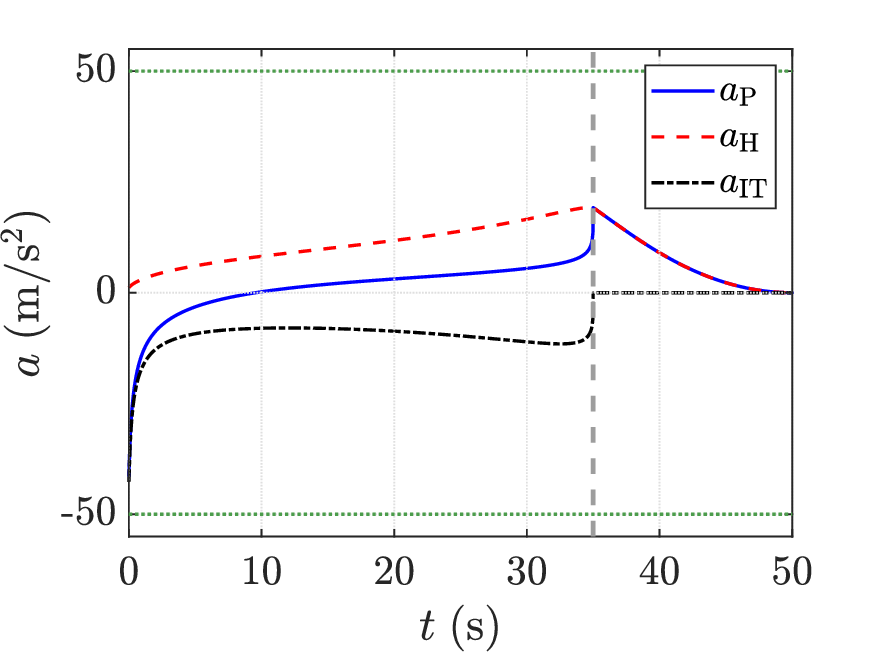}
		\caption{Lateral acceleration profiles.}
		\label{fig:devia_aP}
	\end{subfigure}
	\hspace{0.04\textwidth}
	\begin{subfigure}[t]{0.43\textwidth}
		\centering
		\includegraphics[width=\linewidth]{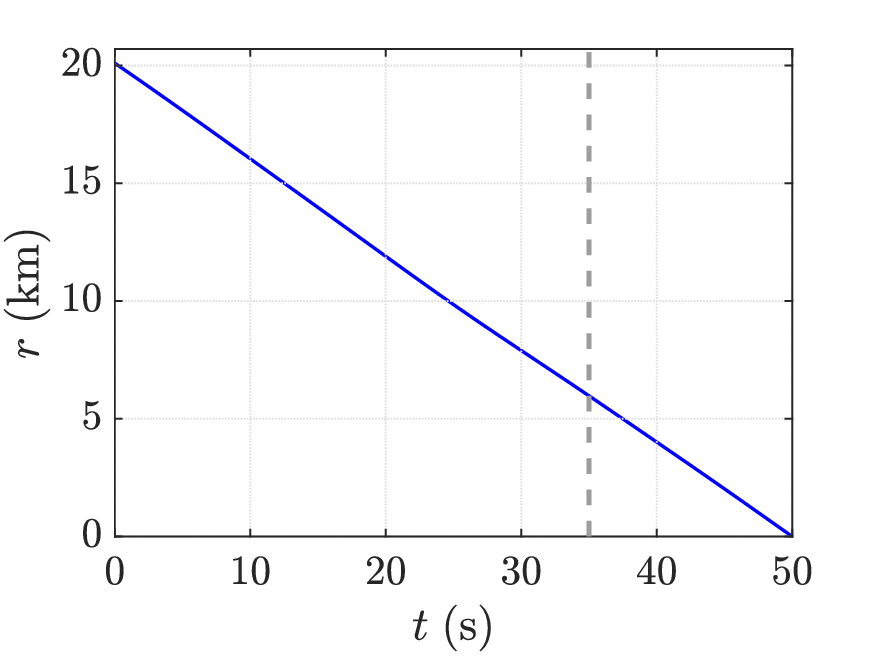}
		\caption{Relative range profile.}
		\label{fig:devia_range}
	\end{subfigure}
	
	\begin{subfigure}[t]{0.43\textwidth}
		\centering
		\includegraphics[width=\linewidth]{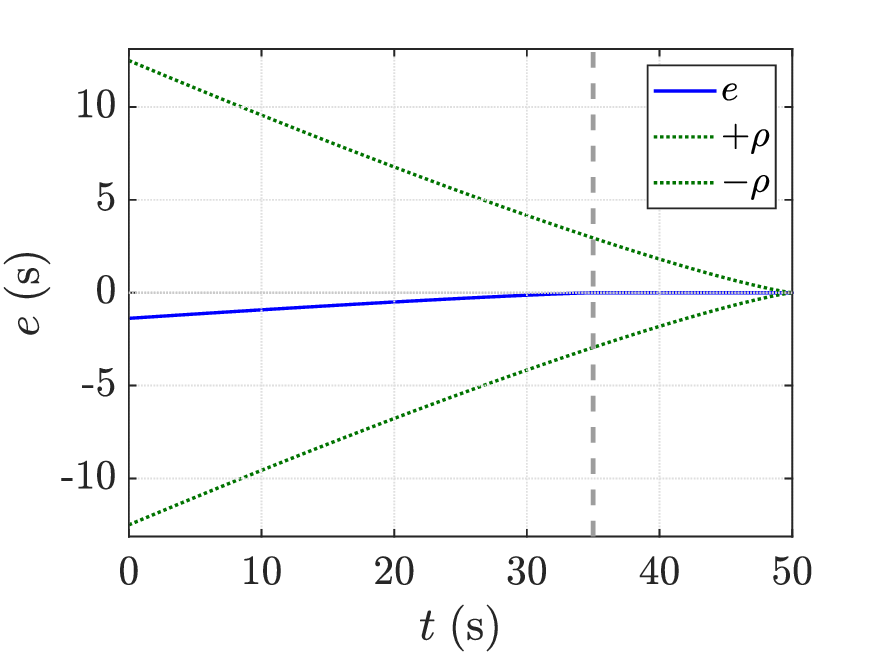}
		\caption{Timing error and prescribed funnel.}
		\label{fig:devia_errors}
	\end{subfigure}
	\hspace{0.04\textwidth}
	\begin{subfigure}[t]{0.43\textwidth}
		\centering
		\includegraphics[width=\linewidth]{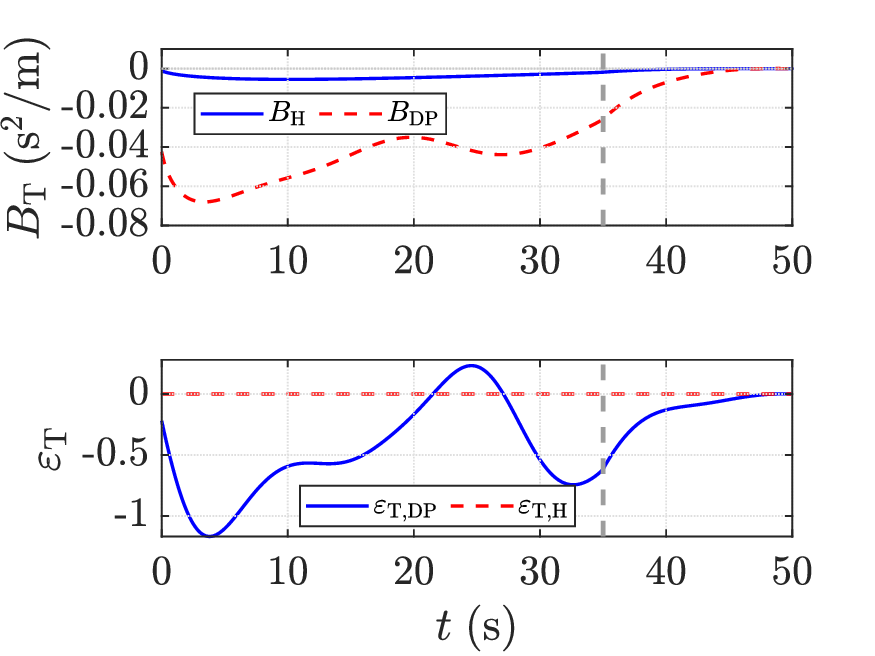}
		\caption{Timing effectiveness and predictor defect.}
		\label{fig:devia_bT}
	\end{subfigure}
	
	\caption{Performance under deviated-pursuit guidance.}
	\label{fig:devia_sim}
\end{figure*}

For the stationary-target case, the interceptor travels at $V_{\mathrm{P}}=300$~m/s with navigation constant $N=3$. The initial geometry is $r(0)=12$~km, $\theta(0)=0^\circ$, and $\sigma(0)=60^\circ$. The desired impact and handover times are $t_f=55$~s and $t_h=35$~s, respectively, with prescribed-time gain $c=2.5$ and acceleration limit $a_{\max}=30$~m/s$^2$. The interceptor reaches the target at the prescribed terminal time, with the range decreasing continuously to zero. The exact homing-compatible time-to-go $T_{\mathrm{H}}$ approaches $t_f-t$, while the large-heading approximation $T_{\mathrm{LPN}}$ initially differs from $T_{\mathrm{H}}$ because of the large heading error and converges toward it as alignment develops. The timing error starts inside the funnel with $\rho(0)=12.5$~s, remains within $\pm\rho(t)$, and approaches zero before handover. Before $t_h$, the bias $a_{\mathrm{IT}}$ modifies the baseline command $a_{\mathrm{H}}$, while $a_{\mathrm{P}}$ remains within $\pm a_{\max}$; after handover, $a_{\mathrm{IT}}$ vanishes and $a_{\mathrm{P}}=a_{\mathrm{H}}$. As collision-course alignment develops, $B_{\mathrm{H}}$ and $B_{\mathrm{LPN}}$ decrease, while the initially nonzero defect $\varepsilon_{\mathrm{T},\mathrm{LPN}}$ approaches the exact-map reference $\varepsilon_{\mathrm{T},\mathrm{H}}=0$.

For the deviated-pursuit case, the interceptor and target travel at $V_{\mathrm{P}}=300$~m/s and $V_{\mathrm{T}}=120$~m/s, respectively. The initial range and LOS angle are $r(0)=20.10$~km and $\theta(0)=5.71^\circ$, respectively, with initial headings $\gamma_{\mathrm{P}}(0)=10^\circ$ and $\gamma_{\mathrm{T}}(0)=150^\circ$. The target follows the acceleration history $a_{\mathrm{T}}(t)=6\sin{(0.20t)}+3\sin{(0.43t+0.5)}$~m/s$^2$. We set $t_f=50$~s, $t_h=35$~s, $c=1.2$, and $a_{\max}=50$~m/s$^2$.
The exact coordinate $\calT_{\mathrm{H}}$ is again computed as the numerical first hitting time in \eqref{eq:baseline_hitting_time}, using the baseline closed-loop flow and the known model of the target's maneuver. The deviated-pursuit expression is assessed through the predictor defect and timing effectiveness. The trajectories reach interception at the prescribed time, and the range decreases continuously toward zero. The coordinate $\calT_{\mathrm{H}}$ approaches $t_f-t$, while the initially different predictor $\calT_{\mathrm{DP}}$ approaches the terminal profile as the engagement evolves. The timing error remains between $\pm\rho(t)$ and approaches zero by handover. The preterminal bias modifies $a_{\mathrm{H}}$ within the total acceleration bounds $\pm a_{\max}$; at $t_h$, the bias vanishes and $a_{\mathrm{P}}$ coincides with the baseline command. Throughout the active correction interval, $B_{\mathrm{H}}$ remains nonzero and then decreases as collision-course alignment develops. The deviated-pursuit coefficient $B_{\mathrm{DP}}$ varies with the relative-motion geometry. The nonzero defect $\varepsilon_{\mathrm{T},\mathrm{DP}}$ quantifies the closed-form predictor's disagreement with the baseline homing flow; the exact coordinate has reference defect $\varepsilon_{\mathrm{T},\mathrm{H}}=0$.

\section{Conclusions and Future Work}\label{sec:conclusions}
We developed a normal-contraction perspective on impact-time guidance that augments a baseline homing command with a scalar timing correction. Under the stated engagement assumptions, terminal calibration and feasible funnel constraints link timing convergence to first capture at the assigned impact time. Declining lateral timing authority favors preterminal alignment, after which an exact nominal predictor allows baseline homing to preserve the assigned arrival time. The reported stationary and modeled maneuvering target cases achieve scheduled interception within the acceleration bounds; approximate predictors exhibit a mismatch that varies with engagement geometry. Future work will incorporate actuator dynamics and sensing constraints and examine three-dimensional cooperative interception.

    \bibliographystyle{IEEEtran}
\bibliography{references}

\end{document}